\documentclass{article}

\usepackage{spconf,amsmath,graphicx}
\usepackage{amsmath,graphicx,amsfonts,amsthm}
\usepackage{epstopdf}\usepackage{xcolor,colortbl}
\usepackage{verbatim}
\usepackage{amsmath,bm}
\usepackage{lipsum}
\usepackage{color}
\usepackage{cite}

\def\bfh{{\mathbf h}}
\def\bfn{{\mathbf n}}

\def\bfs{{\mathbf s}}

\def\cR{{\mathbf R}}
\def\SINR{\mathrm{SINR}}

\def\cC{{\cal C}}
\def\cB{{\cal B}}
\def\cR{{\cal R}}
\def\E{\mathbb{E}}

\newtheorem{theorem}{Theorem}
\newtheorem{lemma}{Lemma}
\newtheorem{corollary}{Corollary}

\newtheorem{definition}{Definition}

\title{RANDOM ACCESS FOR MASSIVE MIMO SYSTEMS \\ WITH INTRA-CELL PILOT CONTAMINATION}
\name{Elisabeth de Carvalho$^*$, Emil Bj{\"{o}}rnson$^\ddagger$, Erik G. Larsson$^\ddagger$, Petar Popovski$^*$}
\address{$^*$Department of Electronic Systems, Aalborg University,  Denmark \\
$^\ddagger$Department of Electrical Engineering (ISY), Link{\"{o}}ping University, Sweden 
\thanks{This work was performed partly in the framework of the Danish Council for Independent Research (DFF133500273), the Horizon 2020 project FANTASTIC-5G (ICT-671660), the EU FP7 project MAMMOET (ICT-619086), ELLIIT, and CENIIT. The authors would like to acknowledge the contributions of the colleagues in FANTASTIC-5G and MAMMOET, as well as the contribution of Dr. Jesper H. S{\o}rensen to the concepts in the paper.}%
}

\begin{document}
\ninept
\maketitle
\begin{abstract}
Massive MIMO systems, where the base stations are equipped with hundreds of antenna elements, are an attractive way to attain unprecedented spectral efficiency in future wireless networks. 
In the ``classical'' massive MIMO setting, the terminals are assumed fully loaded and a main impairment to the performance comes from the inter-cell pilot contamination, i.e., interference from terminals in neighboring cells using the same pilots as in the home cell. However, when the terminals are active intermittently, it is viable to avoid inter-cell contamination by pre-allocation of pilots, while same-cell terminals use random access to select the allocated pilot sequences. This leads to the problem of \emph{intra-cell} pilot contamination. We propose a framework for random access in massive MIMO networks and derive new uplink sum rate expressions that take intra-cell pilot collisions, intermittent terminal activity, and interference into account. We use these expressions to optimize the terminal activation probability and pilot length.
\end{abstract}

\section{Introduction}

In massive multiple-input multiple-output (MIMO) systems the base station (BS) has a large number of antennas, which can be used to create statistically stable and strong spatial beams to the terminals, which are in effect hardened communication channels with negligible small-scale fading. The beamforming depends critically on the channel estimation carried out at the BS, based on the pilot signal sent by each of the terminals that intend to communicate with the BS in the uplink (UL) or downlink (DL). The channel estimation process is deteriorated if the transmission of the pilot sequence is interfered by a concurrent transmission from a terminal that uses the same pilot sequence. If the concurrent transmission (or several of them) are coming from terminals associated with different BSs, then collision occurs, which is the well-known pilot contamination problem \cite{Jose2011b}. The main line of work on massive MIMO, starting from \cite{Marzetta2010a}, has assumed that all terminals in a given cell use orthogonal pilots and analyzed the system performance under inter-cell pilot collisions.

In this paper we reverse this ``classical'' assumption and assume that the interference from other cells is negligible, due to natural separation or orthogonal resource allocation. In contrast, we notice that it can happen that two terminals in the same cell choose the same pilot sequence, leading to \emph{intra-cell pilot collision} or \emph{intra-cell pilot contamination}. This is justifiable in scenarios where the terminals have intermittent traffic \cite{BjornsonGC2015}, such that the number of terminals $K$ associated to a BS is much larger than the number that is active at a certain instant. In such a setting, the number of pilot sequences should closely match the expected number of active terminals rather than the total number of terminals $K$. 

This model is relevant in the classical scenario of random access, in which the terminals are not fully loaded with traffic and there is uncertainty at the BS regarding which terminals have data to send at a given time, such that no scheduling can be applied. In the context of the emerging 5G scenarios, the model covers the crowded scenarios (e.g., stadium) and hotspots \cite{METIS_D11_short}. Another emerging scenario associated with this traffic pattern is where a crowd of sensors occasionally and at random time instants want to transmit data to a common access point. Typically, this
transmission is rather insensitive to delays, the rates are low, and
the uplink power budget is extremely limited. Examples include massive sensor
telemetry in IoT and massive M2M in 5G, where many sensors take measurements 
that need be reported to a fusion center. Note that wireless sensor networks often rely on multi-hop transmissions and path diversity to combat fading towards the sink node. The hardened channels of massive MIMO obviate the need for multi-hop transmissions and provides the spatial diversity required to handle massive traffic loads.

In the approach proposed in this paper, the channels are estimated from uplink pilots
every time the terminal (sensor) transmits.  A data codeword is sent over multiple time slots. In each time slot, each active terminal selects (pseudo-)randomly a pilot from a predetermined pilot codebook and, during the rest of the slot, it sends a part of the data codeword. It can be considered that the terminal performs \emph{pilot hopping} over multiple slots and the hopping sequence can be used to identify the terminal and appropriately merge and decode the parts of its codeword at the BS. This approach is suitable for low-power terminals (by virtue of the large array gain 
of a massive array) and is scalable with respect to the number of antennas at the BS. 
Scalability with respect to the number of terminals is determined by the channel coherence (e.g., determined by the mobility and delay spread) and the activity level of the terminals.

While massive MIMO is a fairly mature research topic \cite{Larsson2014a,Huh2012a,Ngo2013a,Bjornson2016b,Hoydis2013a,Bjornson2016a}, the existing results on uplink capacity analysis in the literature \cite{Ngo2013a} assumes full data buffers and are not applicable to the case we study here. Some preliminary results on the effect of intermittent terminal activity can be found in \cite{BjornsonGC2015}. Here, we take this work one step further and consider a full-blown setup that allows for uncoordinated pilot use and hence fully uncoordinated operation. The aspects of random access in massive MIMO have been recently considered in \cite{SDP2014}, where the use of coded access and successive interference cancellation are considered in the context of massive number of antennas.

\section{Random access and system model}

As described in the introduction, there are important practical scenarios where the pilots used in the home cell are not exposed to pilot contamination from other cells. We therefore consider the UL of a single-cell multi-user massive MIMO system with random access from a large set of intermittently active terminals. The BS is equipped with $M$ antennas and can serve a maximal number of $K$ terminals.
%
The channel coherence interval is  $\tau_c$ symbols long.
A total number of $\tau_p$ orthogonal sequences are available, denoted as 
 $\{ \bfs_1 , \bfs_2 , \dots, \bfs_{\tau_p} \}$, where each sequence is $\tau_p$ symbols long and $\tau_p < \tau_c$.
Moreover, we have $K \gg \tau_c$ so the BS does not have the resources to dedicate pilots to particular terminals.
The duration of a UL time slot  $\tau_u$  is smaller or equal to the coherence interval $\tau_c$. 

The structure of a UL transmission frame is displayed in Fig.~\ref{fig:Model3}. 
In each UL time slot, each terminal decides randomly whether or not to transmit. The decision is made independently from the other terminals and the transmission activation probability $p_a$. The terminal selects a pilot sequence uniformly at random from the pool of $\tau_p$ available pilot sequences. Collisions can thus happen in the \emph{pilot domain}, i.e., among contending terminals that send to the same BS. In each UL slot, the pilot phase is followed by a \emph{data phase}, i.e., transmission of a part of a codeword. The whole codeword is sent over multiple slots. 
For an asymptotically large number of time slots, the whole codeword is affected by an asymptotically large number of channel fading realizations, pilot collisions, and interference events. Relying on the ergodicity of such a process, we characterize the performance through a lower bound on the ergodic capacity.

In random access, the BS does not know a priori which terminals that transmit in a given time slot, or which pilot that a terminal has selected in that slot. In principle, the terminals could select the pilot hopping according to a unique, predefined pseudorandom 
sequence, called \emph{pilot-hopping sequence} here. The BS then knows in advance the pilot-hopping sequence of all potential transmitters, such that it can buffer the information from different slots and run a correlation decoder across the slots in order to find out which pilot-hopping sequences have been activated. Here we do not treat the details of such a procedure and leave it for future work. Instead, we assume that the BS can determine exactly when the terminals were active. The main goal of this paper is to establish   a performance bound for such communication systems, and, based on this bound,    optimize   the activation probability $p_a$ and the number of pilot sequences  $\tau_p$ 
for given system parameters, i.e., the total number of terminals $K$, the uplink time slot duration $\tau_u$ and the number of BS antennas $M$.

\begin{figure}[!t]
\centering
\includegraphics[width=8.3cm]{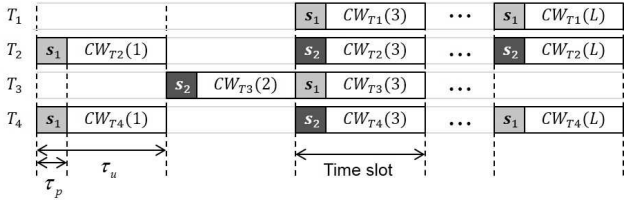} \vspace{-1mm}
\caption{Illustration of the transmission frame. In this example, four terminals $\{T_1,T_2, T_3,T_4\}$ and two mutually orthogonal pilot sequences  $\{s_1,s_2\}$ are considered. Transmission of a codeword is done over multiple channel fades, which enables averaging over noise, channel fades, and pilot collision events.}
\label{fig:Model3}
\end{figure}

A block fading model is adopted where a channel realization is constant across a time slot duration and changes independently from slot to slot.
The channels are narrowband and thus the channel response between the BS and terminal $j$ is described by an $M \times 1$ channel vector $\bfh_j$. The channel realizations are modeled as circularly symmetric complex Gaussian distributed, 
$\bfh_j \sim {\cal C \cal N}(\mathbf{0}, \beta_j \mathbf{I}_M)$. The variance
${\beta_j}$ reflects the path loss, shadowing, received noise power, and the effects of transmit power control at the terminal.
More specifically, statistical power control is performed at the terminals so that $\beta_j$ fluctuates around a nominal value $\overline{\beta}$ according to 
$\beta_j = \overline{\beta} +  v$, where $v$ is modelled as a uniformly distributed random variable between and  $-\alpha \overline{\beta}$ and    
$\alpha \overline{\beta}$, where  $\alpha$ is a constant smaller than 1. 
The normalized $M \times 1$ noise vector $\bfn$ is modelled as 
$\bfn \sim {\cal C \cal N}(\mathbf{0}, \mathbf{I}_M)$, thus the median SNR at each antenna of the BS is $\rho = \overline{\beta}$. 

We use  $(\cdot)^*$,  $(\cdot)^T$, $(\cdot)^H$, $\E[\cdot]$  to denote complex conjugation, transpose, Hermitian transpose, and the expected value of a random variable, respectively.
$\cB_{r,n,p} = (\!\! \begin{array}{c} n \vspace{-1mm}\\ r \end{array}\!\!) p^r (1-p)^{n-r}$ is the probability mass distribution of a binomial distribution with parameters $r$, $n$, $p$.

\section{Lower bound on the Uplink sum rate}

We present three performance expressions that are lower bounds on the ergodic sum rate. The first bound, $\cR_1$, is tight but necessitates Monte-Carlo simulations to be evaluated.
The second  bound, $\cR_2$,
 does not require Monte-Carlo simulation but its tightness depends on the distribution of the parameters $\{\beta_j\}$. 
The third bound, $\cR_3$, is relatively loose, but analytically simple and follows the variations of the ergodic sum rate well.
This bound is used in this paper to optimize the pilot length $\tau_p$ and the activation probability $p_a$.
 
The bounds account for channel estimation errors due to the receiver noise and pilot collisions.
In a given time slot, we assume that the pilot sequence selected by an active terminal $k$ is detected and the channel $\bfh_k$ is estimated using the conventional MMSE estimator \cite{Ngo2013a}. This estimate is used at the BS for maximum ratio combining (MRC) during the data phase. Notice that MRC is an attractive scheme in massive MIMO due to its low computational complexity and near-optimality when $M$ is large \cite{Larsson2014a}.

We denote by $\cC_0$   the set of colliders to one given terminal $0$ (i.e., the active terminals that use the same pilot sequence).
The index $0$ is generic and the results do not depend on it. Due to space limitations, we describe the methodology used to derive the bounds without going into the exact details.
\begin{figure*}[t]
\begin{equation}
\underline{\SINR}_1 = 
\frac{ \tau_p (M-1) \beta_0^2 }
{
\tau_p (M-1) \sum_{j \in \mathcal{C}_0} \beta_j^2 + 
\sum_{i \in  \{0, \cC_0  \} }\beta_i (1 + \tau_p \sum_{j \in {\cal C}_i}^{ }  \beta_j )+ 
( 1+ \sum_{i \notin \{0, \cC_0 \}} \beta_i ) 
({1 +  \tau_p \sum_{i\in {\{0, {\cal C}_0} \}}^{ }  \beta_i } )
}
\label{eq:snrk2}
\end{equation}
\hrule
\end{figure*}
\begin{theorem} \label{theorem:first-theorem}
Assuming MRC at the BS, a lower bound on the ergodic sum rate is
\begin{eqnarray}
\cR_1 = 
 \sum_{K_a=1}^K  p(K_a) \, K_a  \sum_{c=0}^{K_a-1}   p(c|K_a) \;
\E_{\beta} \left[ R_1(\cC_0|K_a)  \right]
\label{eq:Rate1}
\end{eqnarray}
where 
\begin{eqnarray}
\E_{\beta} \left[R_1 (\cC_0|K_a) \right] =  \frac{\tau_u - \tau_p}{\tau_u}\, \E_{\beta} \left[ \log_2(1+ \underline{\SINR}_1) \right]
\label{eq:Rk}
\end{eqnarray}
is a lower bound on the ergodic capacity of terminal $0$ conditioned on a collider set $\cC_0$ and $K_a$ active terminals. The expectation is taken with respect to $\beta_j$ and $\underline{\SINR}_1$ is given by
\begin{eqnarray}
\frac{
\displaystyle{ (M-1) \sigma^2_{\hat{\mathbf{h}}_0} \beta_0^2}
}
{
\displaystyle{
(M-1) \sigma^2_{\hat{\mathbf{h}}_0}  \sum_{j \in \mathcal{C}_0}  \hspace{-1mm} \beta_j^2 + 
 \beta_0^2 (  \hspace{-2mm} \sum_{j \in  \{0, \cC_0  \} } \hspace{-1mm} \sigma^2_{\epsilon_j}   + 
 \hspace{-1mm} \sum_{j \notin \{0, \cC_0  \}}  \hspace{-1mm} \beta_j  + 1 ) }
 }.
\label{eq:snrk}
\end{eqnarray}
%
Note that $p(c|K_a) = \cB_{c, K_a-1,1/\tau_p}$ is the probability of having $c$ colliders to terminal $0$ and on that there are $K_a$ active terminals. $p(K_a)=\cB_{K_a,K,p_a}$ is the probability of having $K_a$ active terminals out of $K$.
$\sigma^2_{\epsilon_j}$ and   $\sigma^2_{\hat{\mathbf{h}}_0}$ are the variance of the channel estimation error and the channel estimate for 
terminal $j$ and $0$. 
\end{theorem}
\begin{proof}
The derivation of~(\ref{eq:snrk}) follows~\cite{Ngo2013a,Bjornson2016a} where the essential ingredient is to treat interference as noise and the use of Jensen's inequality on the function $\log_2(1+1/x)$, which allows averaging over the channel fades of the interferers. This bound is tight thanks to channel hardening. The prelog term $p_a$ in~(\ref{eq:Rk})  comes by accounting for the activity probability of terminal $0$. $\cR_1/K$ is a lower bound on the ergodic capacity of any  given terminal. 
\end{proof}

Replacing the expression of $\sigma^2_{\epsilon_j}$ and  $\sigma^2_{\hat{\mathbf{h}}_0}$  in (\ref{eq:snrk}) with the exact expressions from \cite{Ngo2013a} we obtain (\ref{eq:snrk2}) at the top of the page.
Bound $\cR_1$ requires  Monte-Carlo simulations, while the  bound that is derived next
can be computed numerically without the need for Monte-Carlo simulations.
  Using again Jensen's inequality on $\log_2(1+1/x)$, a lower bound is obtained by taking the expected value of the denominator in (\ref{eq:snrk2}) w.r.t. a) all sets of contaminators to terminal $0$ , b) the parameter $\beta$ associated to the terminal of interest. 
%
\begin{corollary}
Assuming MRC at the BS, a lower bound on the ergodic sum rate is
\begin{eqnarray}
\cR_2 = \sum_{K_a=1}^K  p(K_a) \, K_a\sum_{c=0}^{K_a-1}   p(c|K_a)   
R_2(c|K_a) 
\label{eq:Rate2}
\end{eqnarray}
where
\begin{eqnarray}
R_2(c|K_a) =  \frac{\tau_u - \tau_p}{\tau_u}\, \log_2(1+ \underline{\SINR}_2)
\label{eq:Rk1}
\end{eqnarray}
and $\underline{\SINR}_2$ is shown in (\ref{eq:snrb}) at top of the next page. Note that $p(c|K_a)$ and   $p(K_a)$ are defined in Theorem \ref{theorem:first-theorem}. 
\end{corollary}
In the expression (\ref{eq:snrb}), we have introduced the notations 
$\overline{\beta^2} = \E [\beta^2 ]$, $\overline{\beta^{-1}}=\E [\beta^{-1} ]$, $\overline{\beta^{-2}}=\E [\beta^{-2} ]$ which are assumed to exist. The existence of a closed form expression depends on the  distribution model of $\{\beta_j\}$.
%
\begin{figure*}[t]
\begin{equation}
%
\resizebox{.95\hsize}{!}{
$\underline{\SINR}_2 \!= \!
\frac{\tau_p\left[M-1\right] }{ \tau_p  \left[ M\!-1\right]  c  
\overline{\beta^2}\,\overline{\beta^{-2}}
+
 \overline{\beta^{-1}}\left[ 1 +\tau_p c \bar{\beta} \right] +   c \overline{\beta} \left[ \overline{\beta^{-2}} +
 \tau_p \overline{\beta^{-1}}  + 
 \tau_p  \overline{\beta^{-2}} \,\overline{\beta} ( c -1)  \right] +
 \left[  1+(K_a -c -1)  \overline{\beta} \right]
 \left[ \overline{\beta^{-2}} +  \tau_p   \overline{\beta^{-1}}\! +\! \tau_p   c  \overline{\beta}  \, \overline{\beta^{-2}}\right]  
}$
}
\label{eq:snrb}
\end{equation}
\begin{equation}
%
\resizebox{.95\hsize}{!}{
$\underline{\SINR}_3 \!= \!
\frac{\tau_p \left[M-1\right] }{\overline{\beta^{-2}}+ 
         \left[M-1\right] \left[p_a K - 1\right]\overline{\beta^2}\,\overline{\beta^{-2}} +
             2 \left[p_a K - 1 \right]  \overline{\beta} \, \overline{\beta^{-2}}\left[  1 -  \overline{\beta} (1- 1/\tau_p)\right]  + 
               \overline{\beta}^2 \,\overline{\beta^{-2}}  p_a^2 K \left[ K-1\right]    + 
          \left[ 1 + (p_a K - 1)\overline{\beta} \, \overline{\beta^{-1}}\right]\left[1+\tau_p\right]
         }$
}
\label{eq:snrb2}
\end{equation}
\hrule
\end{figure*}

Next, we present the final sum rate expression used to optimize the parameters $\tau_p$ and $p_a$. 
In the new  bound $\cR_3$, the expectation is taken in the denominator of (\ref{eq:snrb}) w.r.t.~the distribution of the number of contaminators and the number of active terminals.
The bound $\cR_3$ is relatively loose as compared to $\cR_1$  and $\cR_2$, since it averages over the number of colliders and active terminals in the interference variances.
However, it follows very well their variations and provides very good optimization results, as shown in the numerical results. To evaluate the sum rate, expressions
$\cR_1$  and $\cR_2$ are preferable. 
We denote $\tau_p^o$ and $p_a^o$ as the value of the parameters optimizing  $\cR_3$.
\begin{corollary}
\label{lemma2}
Assuming MRC at the BS, a lower bound on the ergodic sum rate is
\begin{eqnarray}
\cR_3 = p_a K \frac{\tau_u - \tau_p}{\tau_u}\log_2(1+ \underline{\SINR}_3)
\label{eq:Rav2}
\end{eqnarray}
where $\underline{\SINR}_3$ is given in \eqref{eq:snrb2} at the top of the next page.
\end{corollary}
\begin{proof}
In the denominator of $\underline{\SINR}_2$, we take the expected value w.r.t. the probability mass of the binomial distribution $p(c) = p(K_a)p(c|K_a)$. 
More specifically, we take first the expected value of $c$ conditioned on a number of active terminals $K_a$. It is the average number of contaminators to one given terminal and is equal to $(K_a -1 )/\tau_p$. Then, we take the expected value w.r.t. $K_a$, i.e. the average number of active terminals out of $K$ terminals which is equal to $p_a K$. Hence, 
$\E(c)= (p_a K-1)/\tau_p$. 
In (\ref{eq:snrb}), there are no contributions in $c^2$ as they get cancelled out. 
\end{proof}

\section{Sum Rate Scaling Laws}
\label{sec:Scale}


Next, we use $\cR_3$ in (\ref{eq:Rav2}) in order to obtain scaling laws and heuristic parameter selection. Consider asymptotic conditions where $p_a K \gg 1$, $M \gg 1$, $\tau_u \gg 1$ and $\tau_p \gg 1$, which are of interest in massive MIMO systems with a high user load that can lead to pilot collisions. 
{An additional condition is $\overline{\beta^{-2}} \;\overline{\beta}^2 \approx 1$, which we assume in the rest of the paper.}
Keeping the dominant terms in (\ref{eq:snrb2}), $\underline{\SINR}_3$ is approximated as:
\begin{eqnarray}  
\underline{\SINR}_a =     
         \frac{M \tau_p }{
        \overline{\beta^2} \, \overline{\beta^{-2}} M  K p_a +
                \overline{\beta}^2 \overline{\beta^{-2}} p_a^2 K^2     +
         \overline{\beta}\,  \overline{\beta^{-1}}p_a K  \tau_p }.
\label{eq:RateAsym}
\end{eqnarray}
The corresponding sum rate expression gives insights into how the sum rate depends on the various parameters. 
Simulations show that the maximum of the sum rate strongly depends  on the term  $\overline{\beta}^2 \overline{\beta^{-2}} p_a^2 K^2$ in the denominator and much less significantly on the other terms. The heuristic solution presented next is obtained based on this observation. 
\begin{definition}
We define $\tau_p^h$ and $p_a^h$ as
\vspace{-1mm}
\begin{eqnarray}   
 {\tau_p^h}= \frac{\tau_u}{3}\quad \quad
 p_a^h K = \sqrt{\frac{\tau_u M}{3 s_o \overline{\beta}^2 \overline{\beta^{-2}} }}
\label{eq:Solheu}
\end{eqnarray}
where $s_0\approx 3.92$ is the solution of $\log(1+x) =2 \frac{x}{1+x}$.  The associated sum rate is equal to
\vspace{-1mm}
\begin{eqnarray}   
\cR_a^h = 
 \sqrt{\frac{\tau_u M}{3 s_o \overline{\beta}^2 \overline{\beta^{-2}} }}\frac{2}{3} \log_2\left(  1+ \underline{\SINR}_a^h  \right).
\end{eqnarray}

\vspace{-1mm}
\noindent
where  $\underline{\SINR}_a^h$ is the value of  $\underline{\SINR}_a$ in (\ref{eq:RateAsym}) at $(\tau_p^h, p_a^h)$.
\end{definition}
\begin{proof}
We look for the expression of $\tau_p$ and $p_a K$ maximizing the following rate function: 
\vspace{-2mm}    
\begin{eqnarray}  
R^h =  p_a K (\tau_u - \tau_p)
 \log_2(1+X),  \quad
X =   \frac{M \tau_p }{
                \overline{\beta}^2 \overline{\beta^{-2}} p_a^2 K^2    }.
\end{eqnarray}

\vspace{-2mm}
\noindent 
The partial derivative of $R^h$ are
\vspace{-2mm}
\begin{eqnarray}     
\left\{ 
\begin{array}{lcl}
\frac{\partial R^h}{\partial  \tau_p} & =&
-  
\log_2(1+X) 
+
(\tau_u - \tau_p)
\frac{\partial X}{\partial  \tau_p}
\frac{1}{1+X},  \\
\frac{\partial R^h}{\partial p_a} & =& \log_2(1+X) 
+ p_a 
\frac{\partial X}{\partial  p_a}
\frac{1}{1+X}.
\end{array}
\right.
\end{eqnarray}

\vspace{-1mm}
\noindent
Noting that 
$
p_a  \frac{\partial X}{\partial  p_a} = 
-2 {X} $ and 
$
 \frac{\partial X}{\partial \tau_p} = 
\frac{1}{\tau_p} {X}
$, we obtain

\vspace{-3mm}
\begin{eqnarray}     
\left\{ 
\begin{array}{l}
1+\frac{1+X}{X} \log(1+X) =
 \frac{\tau_u}{\tau_p} , \\
\frac{1+X}{X}   \log(1+X) 
=2. \\
\end{array}
\right.
\end{eqnarray}
\vspace{-2mm}
From those equations, we obtain (\ref{eq:Solheu}).
\end{proof}

Based on this heuristic parameter selection, we obtain the following scaling behaviors.
\begin{lemma}
\label{lemma3}
Assuming $\tau_u \gg 1$, $\tau_p \gg 1$,  $M \gg 1$ and $p_a K \gg 1$,  the following asymptotic results hold for 
$\tau_p^h$ and $p_a^h K$ in (\ref{eq:Solheu}):
\begin{enumerate}
\item $M \gg \tau_u$: $\underline{\SINR}_a^h$ scales as  $\sqrt{\tau_u/M}$ and $\cR_a^h$ scales as  $\tau_u$.
\item $M \ll \tau_u$: $\underline{\SINR}_a^h$ scales as  $\sqrt{M/\tau_u}$ and $\cR_a^h$ scales as $M$.
\item $M \sim \tau_u$: $\underline{\SINR}_a^h$ tends to a constant value 
 and $\cR_a^h$ scales as  $\sqrt{\tau_u M}$.
\end{enumerate}
\end{lemma}
\begin{proof}
Substituting the expressions (\ref{eq:Solheu}) in  (\ref{eq:RateAsym})  leads to those asymptotic results. 
\end{proof}

The significance of the heuristic solution is two-fold. First, this solution gives a sum rate that is close to the optimal sum rate, which will be illustrated in Section \ref{section:numerical}. Second, it provides quasi-optimal variation laws in all asymptotic regimes. In simulations, a dependence of $\tau^o_p$ on $M$ can be observed but it is weak. Furthermore, it is natural to model the dependence of 
$\tau_p^o$ on $\tau_u$ only and not on $M$ to comply with the constraint $\tau_p^o \leq \tau_u$. 
Examining the laws of variations   $\tau_p \sim O(\tau_u^a)$ and $p_a K \sim O(M^b \tau_u^c)$,
one can easily show  that the choice $a=1, b=c=1/2$ leads to the best scaling laws. 

When $\tau_u$ is the smaller quantity, the average number of active terminals  and the sum rate is limited by $\tau_u$. 
When $M$ is the smaller quantity, their number is limited by $M$. 
When $M$ and $\tau_u$ are comparable, the optimal number of pilot sequences and average number of active terminals  becomes comparable as well. 
In the first two asymptotic modes of Lemma~\ref{lemma3}, the rate of each terminal becomes asymptotically small but the average number of active terminals that the system can accommodate grows faster. 
In the third mode, the rate of each terminal becomes constant while the average number of active terminals increases. The quality of service requirement should dictate which values of $M$ and $\tau_u$ should be selected. 
Note that, the system functions in regimes where the number of average number of active terminals is of same order of number of antennas.

\section{Numerical Results}
\label{section:numerical}

In this section, we illustrate the behavior of the 3 performance bounds as well as the optimal and heuristic solutions. 
The SNR $\rho = \overline{\beta}$ is fixed to 10dB. 
The variation parameter of $\beta_j$  around $\overline{\beta}$ is set to  $\alpha= 0.25$.
The total number of terminals $K$ is equal to 800. 

Fig.~\ref{fig1} displays the variations of $(\tau_p^o, p_a^o K)$  and $(\tau_p^h, p_a^h K)$ as a function of $\tau_u$.
The number of antennas is $M=100$ and $M=400$. We can see that $\tau_p^o$ and $\tau_p^h$ follow a linear variation w.r.t. $\tau_u$, 
but the linear coefficient slightly depends on the value of $M$ for $\tau_p^o$. $p_a^o K$ and $p_a^h K$ both scales as $\sqrt{\tau_u}$. 
For $M=100$, the average number of active terminals is limited by the number of antennas and is smaller than $\tau_p^o$. When $M$ and $\tau_u$ are comparable, the optimal regime involves a comparable number of active terminals and number of pilot sequences. 
Particularly for large values of $M$, the offset between optimal and heuristic solutions becomes large, while the sum rate exhibits a small difference as shown in Fig.~\ref{fig2}. This comes from the fact that the region around the optimal solution is quite flat so that such an offset does not have a significant impact.  

In Fig.~\ref{fig2}, the performance bounds $\cR_1$, $\cR_2$ and $\cR_3$ are displayed for $M=100$  and $M=400$. Bound $\cR_1$ and $\cR_2$ are almost superposed for $M=100$ while a small gap is visible for $M=400$. 
 A large gap between $\cR_3$ and the other bounds can be observed.
This gap comes
 from the large variations of $\underline{\SINR}_2$ in (\ref{eq:snrb}) w.r.t. the collision events that are averaged out in the denominator of $\underline{\SINR}_2$ to get $\cR_3$. At last, looking at bound $\cR_1$, we see that the heuristic solution exhibits excellent performance.

\begin{figure}[!h]
\centering
\includegraphics[width=6.8cm]{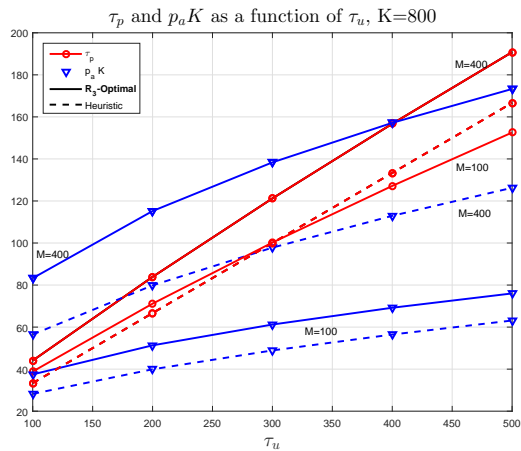}
\vspace{-2mm}
\caption{Optimal $\tau_p^o$ and $p_a^o K$ and heuristic $\tau_p^h$ and $p_a^h K$ as a function of $\tau_u$ for M=100 and M=400.
}
\label{fig1}
\end{figure}

\begin{figure}[!h]
\centering
\hspace{0mm}
\includegraphics[width=6.8cm]{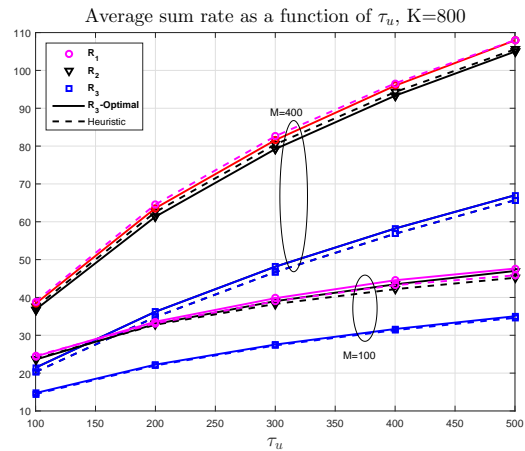}
\vspace{-2mm}
\caption{Performance bounds $\cR_1$,  $\cR_2$,  $\cR_3$ evaluated at $(\tau_p^o, p_a^o K)$ and  $(\tau_p^h, p_a^h K)$
for  M=100  and M=400.
}
\label{fig2}
\end{figure}

\section{Conclusion}
We have considered a communication scenario with massive MIMO and intermittent terminal activity.  In such a setting it is infeasible to allocate orthogonal pilots within a cell and the terminals apply random access to a small common pilot set. On the other hand, the pilot sets allocated to the neighboring sets are orthogonal. This gives a rise to operation that is free from the usual inter-cell pilot contamination and instead leads to intra-cell pilot contamination that occurs as a result of a collision of a random access process. We have provided performance expressions as well as optimization tools that are particularly important for a system where the activity of the terminals and the number of pilots have to obey certain statistical rules.

\bibliographystyle{IEEEbib}
\bibliography{strings,IEEEabrv,refs}

\end{document}